\documentclass[letterpaper,10pt]{article}

\def\tit{Multivariate sparse interpolation using
randomized Kronecker substitutions}

\usepackage[svgnames]{xcolor}

\usepackage[
  pdfauthor={Andrew Arnold and Daniel S. Roche},
  pdftitle={\tit},
  colorlinks,
  linkcolor=DarkBlue,
  citecolor=DarkGreen,
  urlcolor=DarkBlue,
  bookmarksnumbered,
  backref=page,
  pdfpagelabels=true,
]{hyperref}

\usepackage{xstring}
\renewcommand{\backref}[1]{Referenced on %
    \expandarg\StrCount{#1}{,}[\ncommas]%
    \ifthenelse{\ncommas = 0}{page #1}%
    {pages \StrBefore[\ncommas]{#1}{,}\ and\StrBehind[\ncommas]{#1}{,}}%
.}
\newcommand{\doi}[1]{doi: \href{http://dx.doi.org/#1}{\path{#1}}}

\usepackage[vlined,ruled,boxed]{algorithm2e}
\DontPrintSemicolon
\SetKw{break}{break}

\usepackage{amsmath}
\usepackage{amssymb}
\usepackage{mathabx}
\usepackage{mathtools} %
\usepackage{paralist} %

\newcommand{\ZZ}{\ensuremath{\mathbb{Z}}}
\newcommand{\F}{\ensuremath{\mathsf{R}}}

\newcommand{\x}{\ensuremath{\mathbf{x}}}
\newcommand{\alvec}{\ensuremath{\boldsymbol\alpha}}
\newcommand{\btvec}{\ensuremath{\boldsymbol\beta}}

\newcommand{\rr}{\ensuremath{\mathbf{r}}}
\newcommand{\s}{\ensuremath{\mathbf{s}}}
\newcommand{\zvec}{\ensuremath{\mathbf{0}}}
\newcommand{\R}{\ensuremath{\mathsf{R}}}
\newcommand{\dd}{\ensuremath{\mathbf{d}}}
\newcommand{\e}{\ensuremath{\mathbf{e}}}

\newcommand{\HH}{\ensuremath{\mathcal{H}}}
\newcommand{\FF}{\ensuremath{\mathbb{F}}}

\newcommand{\bnd}[2]{\ensuremath{#1\left(#2\right)}}
\newcommand{\oh}[1]{\bnd{O}{#1}}
\newcommand{\softoh}[1]{\bnd{\widetilde{O}}{#1}}
\newcommand{\ceil}[1]{\ensuremath{\left\lceil{#1}\right\rceil}}

\DeclareMathOperator{\lc}{lc}

\title{\tit}
\author{Andrew Arnold\\
\small Cheriton School of Computer Science\\
\small University of Waterloo\\
\small Waterloo, Ontario, Canada\\
\href{https://cs.uwaterloo.ca/~a4arnold/}{\tt a4arnold@uwaterloo.ca}
\and
Daniel S.\ Roche\\
\small Computer Science Department\\
\small United States Naval Academy\\
\small Annapolis, Maryland, USA\\
\href{http://www.usna.edu/cs/roche/}{\tt roche@usna.edu}}

\usepackage[longnamesfirst]{natbib}
\usepackage{amsthm}
\usepackage{todonotes}

\newtheorem{theorem}{Theorem}[section]
\newtheorem{lemma}[theorem]{Lemma}
\newtheorem{cor}[theorem]{Corollary}
\newtheorem{definition}[theorem]{Definition}

\numberwithin{equation}{section}

\begin{document}
\maketitle

\begin{abstract}
  We present new techniques for reducing a multivariate
  sparse polynomial to a univariate polynomial. 
  The reduction works similarly to the classical and
  widely-used Kronecker substitution, except that we choose
  the degrees randomly based on the number of nonzero terms
  in the multivariate polynomial.
  The resulting univariate polynomial often has a significantly
  lower degree than the Kronecker substitution polynomial,
  at the expense of a small number of term collisions.
  As an application,
  we give a new algorithm for multivariate interpolation
  which uses these new techniques
  along with any existing univariate interpolation algorithm.
\end{abstract}

\section{Introduction}

We consider the problem of determining the coefficients and
exponents of an unknown sparse
multivariate polynomial, given a ``black box'' procedure that allows for
its evaluation at any chosen point.

Our new technique is a variant on the classical Kronecker substitution.
Say $f$ is an $n$-variate polynomial with max degree less
than $D$ and coefficients in a ring $\R$. The Kronecker substitution
produces a univariate polynomial $g\in\R[z]$ by substituting powers of
$z$ in the evaluation of $f$:
$$g(z) = f\left(z, z^D, z^{D^2},\ldots, z^{D^{n-1}}\right).$$
This map is invertible because there is a one-to-one correspondence
between terms in $f$ and in $g$, but the price of such 
convenience is an exponential increase in the degree.

For example, consider the following bivariate polynomial:
$$f = 3x^9 y - 2x^5y^4 -y^6 + x^2y^9.$$
The standard Kronecker substitution would be
$$g = f(z, z^{10}) = 3z^{19} - 2z^{45} - z^{60} + z^{92}.$$
Every term in $g$ comes from a single term in $f$, and the original exponents
can be determined by the base-$D$ expansion of exponents in $g$.

Our randomized Kronecker substitution is also a map from multivariate
to univariate polynomials obtained by evaluating 
at powers of a single indeterminate. We choose $n$ integers
$(s_1,\ldots,s_n)$ at random and perform the substitution 
$g(z) = f(z^{s_1}, \ldots, z^{s_n})$. When these integers are not too
large (in particular, if each $s_i < D^{n-1}$), the degree of $g$ will
be less than in the usual Kronecker substitution. 

The price of a decreased degree is that the map is no longer invertible,
for two reasons.
First, we may have two or more distinct terms
in $f$ converge to a single term in $g$. This is called a
\emph{collision}.
The second difficulty is that the original multivariate exponents cannot
be determined directly from the terms in a single substitution $g$.
We will show how performing $O(n+\log \#f)$ such random substitutions
can overcome both difficulties.

In the example above, choose $s_1=5$ and $s_2=2$,
so that
$$g_{5,2}(z) =  f(z^5, z^2) = -z^{12} + z^{28} - 2z^{29} + 3z^{47}.$$
In this case $\#g = \#f$, so there were no term collisions, even
though their order has changed. The advantage is that the degree
is significantly less than that from the usual Kronecker substitution.

Choosing instead $s_1=2$ and $s_2=5$, the result is
$$g_{2,5}(z) = f(z^2, z^5) = 3z^{23} - 3z^{30} + z^{49},$$
which again has a reduced degree, but in this case we have a
\emph{collision}: The two terms $-2x^5y^4-y^6$ in $f$ collided to
produce a single term $-3z^{30}$ under this substitution.

Nonetheless, for the two terms not involved in a collision,
both images can be used to recover the original terms in $f$.
The two terms with coefficient $3$ in
the substitutions are
$3z^{47}$ and $3z^{23}$. This produces the linear system
$$
\begin{bmatrix}5&2\\ 2&5\end{bmatrix}
\begin{bmatrix}u\\ v\end{bmatrix}
=
\begin{bmatrix}47\\ 23\end{bmatrix},
$$
which is solved to reveal 
exponents $u=9$ and $v=1$ of the original term $3x^9y$.

Crucial to our success is determining a suitable set of
integers from which to choose the $s_i$.
Section~\ref{sec:subs} provides
bounds to construct such sets with provably many ``good'' choices
which will not produce many collisions. Recovering the original terms
requires a way to correlate like terms in the images $g$, and
Section~\ref{sec:diverse} shows how a separate randomization 
allows the coefficients to be used to identify like terms in separate
images $g$.
Section~\ref{sec:interp} describes a multivariate interpolation
algorithm that makes use of these randomizations.

\section{Related work}\label{sec:related}

Polynomial interpolation dates from the 18th century, where the goal was
to discover a polynomial approximation to an
unknown function by collecting sufficiently many observations.
It is often beneficial to search for a \emph{sparse}
model, where only a bounded number of terms in the polynomial are
nonzero (see, e.g., \citet{CRT06}).

Our setting is more restrictive in two senses. First, we require the
ability to choose the points at which the unknown polynomial is
evaluated. In this \emph{black-box} interpolation setting, 
the polynomial a procedure which can be probed
with any desired input, at some cost (see, e.g., \citet{GKS90,Man95,GLL09}).

We also require that a sparse polynomial truly exists ``inside the box'',
and all evaluations --- even those with numerical noise --- come from
the same $T$-sparse polynomial. This is different from the general
setting above where the true function need not be $T$-sparse or indeed
a polynomial at all. (Allowing for outliers as in \citet{CKP12} provides
another option.)

Algorithmic progress in sparse interpolation can be divided
into two categories depending on the complexity.
Algorithms with polynomial dependence on the partial degree bounds,
sparsity bound, and number of variables are called \emph{sparse} algorithms
\citep{Zip79,Zip90,HR99}.
Those with only logarithmic dependence on the degree are \emph{supersparse}
and can be
useful even in the case of univariate polynomials \citet{BT88,AKP06}.

As shown in the introduction, Kronecker subtitution \citep{Kro82} can turn
a univariate interpolation into a multivariate one, by
evaluating the multivariate polynomial at high powers of the univariate
interpolation points.  This
is especially effective in conjunction with supersparse
algorithms \citep{Kal10a}.

Zippel's sparse interpolation method 
proceeds one variable at a time and uses randomization to identify 
which portions of the unknown sparse polynomial vanish.
This is another way to turn any univariate algorithm into a
multivariate one \citet{KL03}.

There is recent work on implementations of sparse interpolation
\citet{JM10,HL14} and applications to GCD and factorization
\citep{JM09,BL12}. The Kronecker substitution in particular has been
applied to integer and polynomial multiplication
\citep{Sch82,Har09a}.

Sparse interpolation has a strong connection to the important
theoretical problem of polynomial identity testing (see \citet{SY10}
for a survey).  Of particular
interest is \citet{KS01}, where a similar technique to ours
is used for identity testing and interpolation.
A significant advantage of their algorithms 
is the very small number of random bits required for the computation.

Our randomized Kronecker substitution provides yet
another way to transform any univariate
interpolation algorithm into a multivariate one.
Table~\ref{table:compare} summarizes our contribution as compared to
\citet{KS01}, Zippel's method, and the normal Kronecker substitution.
The table shows, for each method,
how many times a univariate interpolation algorithm must be called, and
the degree of the univariate polynomials that must be interpolated,
relative to a given bound $D$ on the max degree of the unknown
multivariate polynomial, and the number of variables $n$. (The sparsity
bound for the univariate algorithms will always be simply $T$.)

\begin{table}[htbp]\caption{Uni-to-multivariate
interpolation methods\label{table:compare}}
\begin{center}
\begin{tabular}{c|c|c}
& \# of reductions & Degree\\
\hline
Kronecker substitution & 1 & $D^n$ \\
Zippel & $nT$ & $D$ \\
Klivans \& Spielman & $n$ & $O(n^2 T^2 D)$ \\
This paper ($n=2$ case) & $O(\log T)$ & \oh{\sqrt{T} D \log D} \\
This paper ($n\ge 3$ case) & $O(n+\log T)$ & \oh{TD}
\end{tabular}
\end{center}
\end{table}

The total cost of a multivariate interpolation algorithm obtained this
way depends on the choice of univariate interpolation algorithm.
Writing $D_1$ for the univariate degree (the second column of
Table~\ref{table:compare}), there are two options: dense methods
which require $D_1$ black box probes and $\softoh{D_1}$ field operations,
or supersparse methods based on Ben-Or and Tiwari which require 
only $O(T)$ probes and $\softoh{T\log^2 D_1}$ field operations.
In Table~\ref{table:comp2}
we compare two standard approaches
against our own: Kronecker substitution with a supersparse univariate algorithm,
and Zippel's method with a dense univariate algorithm.

\begin{table}[htbp]\caption{Overall multivariate interpolation\label{table:comp2}}
\begin{center}
\begin{tabular}{c|c|c|c}
reduction & univariate alg. &
\# of probes & \# of field ops \\
\hline
Kronecker & supersparse & $2T$ & $\softoh{n^2T\log^2 D}$ \\
Ours & supersparse & $\softoh{nT}$ & $\softoh{n T \log^2 D}$ \\[2 mm]
Zippel & dense & $nTD$ & $\softoh{nTD}$ \\
Ours ($n=2)$ & dense & $\softoh{\sqrt{T} D}$ & $\softoh{\sqrt{T} D}$ \\
Ours ($n\ge 3)$ & dense & $\softoh{nTD}$ & $\softoh{nTD}$
\end{tabular}
\end{center}
\end{table}

\section{Randomized substitutions}\label{sec:subs}

For an unknown polynomial $f \in \F[x_1,\ldots,x_n]$, 
our main technical
contribution is a way of choosing integers
$s_1,\ldots,s_n \in \ZZ$ such that the substitution
$g = f\left(z^{s_1},\ldots,z^{s_n}\right)$
results in a lower degree than the usual Kronecker substitution,
while probably not introducing too many term collisions.

\subsection{Bivariate substitutions}

We begin with the case of $n=2$ variables, where our result is
stronger than the general case and we always choose the random
substitution exponents $s_1, s_2$ to be prime numbers.
Bivariate polynomials naturally
constitute a large portion of multivariate polynomials of interest, and
they correspond to the important case of converting between polynomials
in $\ZZ[x]$ and multiple-precision integers \citep{Sch82,Har09a}. 

Throughout this subsection, we assume $f\in\F[x,y]$ is an unknown
bivariate polynomial, written as
\begin{equation}\label{eqn:f-bivar}
  f = a_1 x^{u_1}y^{v_1} + a_2 x^{u_2}y^{v_2} + \cdots + a_t
  x^{u_t}y^{v_t}.
\end{equation}
We further assume upper bounds $D_x, D_y$ on the $\deg_x f$ and
$\deg_y f$, respectively, and $T \ge t$ on the number of nonzero terms
$\#f$.

The general idea here is to perform the substitution
\begin{equation}\label{eqn:rk-bivar}
g(z) = f(z^p, z^q)
\end{equation}
for random chosen prime numbers $p$ and $q$. 
We want to choose $p$ and $q$ as small as possible, so as to minimize
$\deg g$, but large enough so that there are not too
many collisions.

Our approach to choosing primes is based on the following technical lemma, which
shows how to guarantee a high
probability success while minimizing the degree of $g$.

\begin{lemma}\label{lem:lambda-pq}
  Let $f\in\F[x,y]$ with partial degrees less than $D_x, D_y$ and at most
  than $T$ nonzero terms, 
  $0 < \mu < 1$ be a chosen bound on the probability of failure,
  and $1\le i \le T$ be the index of a nonzero term in $f$. Define
  $$
    B = \frac{25 \left(T-1\right) \ln D_x \ln D_y}{9 \mu},
  $$
  $$
    \lambda_p = \max\left(20.5, \sqrt{\frac{B D_y}{D_x}}\right)
      \quad \text{and} \quad
    \lambda_q = \max\left(20.5, \sqrt{\frac{B D_x}{D_y}}\right)
  $$

  By choosing primes $p,q$ uniformly at random from the ranges
  $[\lambda_p,2\lambda_p]$ and
  $[\lambda_q,2\lambda_q]$ respectively, 
  the probability that the $i$th term of $f$ collides
  in $f(z^p, z^q)$ is less than $\mu$.
\end{lemma}

\begin{proof}
  The primes $p$ and $q$ constitute the random choices in this discussion.
  We say a pair $(p,q)$ is \emph{bad} if the term $a_ix^{u_i}y^{v_i}$ collides
  with any other term in $f(z^p, z^q)$. We will show
  that the number of bad pairs
  $(p,q)$ is at most $\mu$ times the total number of prime pairs that
  could be chosen.

  We begin with a simple lower bound on the latter quantity.
  Applying equation (3.8) in
  \citet{RS62} twice shows that the total number of ordered
  pairs $(p,q)$ is at least
  \begin{equation}\label{eqn:pq-rhs}
    \frac{9\lambda_p \lambda_q}{25\ln\lambda_p \ln\lambda_q}.
  \end{equation}

  Next we obtain an upper bound on
  the number of bad pairs by counting the total number of times the $i$th
  term collides in
  every possible $f(z^p,z^q)$.

  Observe that for any set $S$ of nonzero integers, and any bound $\lambda$,
  the number of times any prime $p \ge \lambda$ divides any of the
  integers in $S$ is at most
  \begin{equation}\label{eqn:div-count}
    \frac{|S| \ln\max(S)}{\ln \lambda}.
  \end{equation}

  A collision with the $i$th term occurs whenever
  $u_i p + v_i q = u_j p + v_j q$ for some $j\ne i$, $1\le j \le T$.
  This happens only when $p\divides(v_i-v_j)$ and $q\divides(u_i-u_j)$. Since
  $i\ne j$, these exponent differences $(u_i-u_j)$ and $(v_i-v_j)$ cannot
  both be zero. Furthermore, if one exponent difference is zero, then the
  collision can never occur.

  Therefore all collisions occur at indices in the set
  $$J = \{j \mid 1\le j \le T \text{ and } (u_i-u_j)(v_i-v_j)\ne 0\}.$$
  The total number of times the $i$th term collides in any $f(z^p,z^q)$
  is equal to the sum over all $j\in J$ of the number of pairs $(p,q)$
  such that $p\divides(v_i-v_j)$ and $q\divides(u_i-u_j)$.

  Now define, for each $p$, the subset of \emph{possible} collision
  indices as 
  $J_p = \{j\in J \mid p \text{ divides } (v_i-v_j)\}.$

  As each $(v_i-v_j)<D_y$, we have from \eqref{eqn:div-count} that
  $$\sum_{p\ge\lambda_p} \#J_p \le (T-1)\ln D_y/\ln \lambda_p.$$

  For each prime $p$, the total number of times the $i$th term
  collides is the number of indices $j\in J_p$ such that
  $q \divides (u_i - u_j)$. As each of these differences is
  less than $D_x$, using \eqref{eqn:div-count} again this sum
  is at most $(\#J_p \cdot \ln D_x)/(\ln \lambda_q)$.

  Therefore the total number of times the $i$th term collides is
  at most
  \begin{equation*}%
    \sum_{p\ge\lambda_p}\frac{\#J_p \ln D_x}{\ln \lambda_q}
    \le \frac{(T-1)\ln D_x \ln D_y}{\ln \lambda_p \ln \lambda_q}.
  \end{equation*}

  Using the definition of $B$ and the observation that
  $\lambda_p \lambda_q = B$, we can rewrite this bound as
  $$
    \frac{9\mu}{25\ln \lambda_p \ln \lambda_q} B = 
    \mu \frac{9\lambda_p \lambda_q}{25 \ln\lambda_p \ln\lambda_q},
  $$
  which is exactly $\mu$ times \eqref{eqn:pq-rhs}. Hence the probability
  of choosing a bad pair $(p,q)$ 
  is at most $\mu$.
\end{proof}

Choosing primes $p,q$ from such sets provides a good bound on the
degree of the resulting polynomial $g$.

\begin{cor}\label{cor:deg-bivar}
Let $f\in\F[x,y]$ with partial degrees less than $D_x, D_y$ and
at most $T$ nonzero terms.

Then for any constant error probability $0<\mu<1$, and
primes $p,q$ chosen randomly as in Lemma~\ref{lem:lambda-pq},
the substitution polynomial 
$g(z) = f(z^p, z^q)$
has degree at most 
$$O(\sqrt{T} \sqrt{D_x D_y} \log (D_x D_y)).$$
\end{cor}

The degree of a standard Kronecker substitution is $D_x D_y$.
Because $T$ is always less than this, the randomized substitution
will never have degree
more than a logarithmic factor larger than $D_x D_y$. The benefit
comes when the polynomial is sparse, i.e., $T \ll D_x D_y$, in which
case the randomized substitution has much smaller degree, albeit at
the expense of a few collisions.

\subsection{Multivariate substitutions}

When $f$ has at least 3 variables, the preceding  analysis 
no longer applies. The new difficulty is that potentially-colliding
terms could have exponents that differ in two \emph{or more} variables, meaning
that the simple divisibility conditions are no longer sufficient to identify
every possible collision. Consequently, our randomly-chosen exponents
in this case will be somewhat larger, and not necessarily prime.

For this subsection,
$f \in \F[x_1,\ldots,x_n]$ is an unknown $n$-variate polynomial, written as
\begin{equation}\label{eqn:f-mvar}
  f = a_1 \x^{\e_1} + a_2 \x^{\e_2} + \cdots + a_t \x^{\e_t},
\end{equation}
where $\x = (x_1,\ldots,x_n)$ and each $\e_i \in \ZZ^n$.
$D$ and $T$ are upper bounds on the max degree and number of terms in
$f$.

Our general approach here is to choose a random vector
$\s = (s_1,s_2,\ldots,s_n)$ of integers below a certain bound,
and then perform the substitution
$g(z) = f(z^{s_1}, \dots, z^{s_n})$.

The following lemma, similar in purpose to Lemma~\ref{lem:lambda-pq},
shows how large the integers in $\s$ must be in order to guarantee 
a small likelihood of collisions.

\begin{lemma}\label{lem:lambda}
  Let $f\in\F[x_1,\ldots,x_n]$ with max degree less than $D$ and 
  at most $T$ nonzero terms,
  $0 < \mu < 1$ be a chosen bound on the failure probability,
  and $1\le i \le T$ be the index of a nonzero term in $f$. 
  Define $\lambda$ to be the least prime number satisfying
  $\lambda \ge T/\mu$.

  If integers $s_1,\ldots,s_n$ are chosen uniformly at random
  from $[0,\lambda-1]$, then 
  the probability that the $i$th term of $f$ collides in 
  $f(z^{s_1},\ldots,z^{s_n})$ is less than $\mu$.
\end{lemma}
\begin{proof}
  Adopt the notation $\FF_\lambda$ for the finite field with $\lambda$
  elements, which we will represent as $\ZZ/\lambda\ZZ$. Write
  $\s = (s_1,\ldots,s_n)$ for the randomly-chosen vector in
  $\FF_\lambda^n$.

  Now let $1\le j \le T$, $j\ne i$, and consider the $j$th term of $f$.
  Writing $\e_i, \e_j$ for the exponent vectors of these terms as in 
  \eqref{eqn:f-mvar}, define $\dd_j = \e_i - \e_j$, which cannot be
  the zero vector as $i$ and $j$ are distinct terms.
  We see that these two terms collide in the
  substitution $f(z^{s_1},\ldots,z^{s_n})$ if and only if
  $\dd_j \cdot \s = 0$.

  Define $\ell \ge
  0$ such that $\lambda^\ell$ is the largest power of $\lambda$ that
  divides every entry in $\dd_j$, and write $\dd_j' =
  \dd_j/\lambda^\ell$. This means that $\dd_j' \in \ZZ^n$
  and $\dd_j' \bmod \lambda \ne \zvec$. Furthermore, $\dd_j\cdot\s = 0$
  if and only if $\dd_j' \cdot \s = 0$.

  Now, if $\dd_j' \cdot \s = 0$, then this also holds modulo $\lambda$,
  so $\s$ must lie in the $(n-1)$-dimension null space of
  $\dd_j' \bmod \lambda$, call it $W$, where $W \subset \FF_\lambda^n$.
  The probability that $\s \in W$ is $1/\lambda$, and therefore the
  probability that terms $i$ and $j$ collide is at most
  $1/\lambda$ as well.

  Since there are $T-1$ terms that the $i$th term could collide with, 
  i.e., $T-1$ choices for $j$, the
  probability that the $i$th term collides with any other term is at
  most $(T-1)/\lambda$. From the definition of $\lambda$, this is less than $\mu$.
\end{proof}

Due to Bertrand's postulate, we will have
$\lambda < 2T/\mu$.
The following corollary shows how this bound on the size of entries in
the randomly-chosen $\s$ affects the degree of the reduced univariate
polynomial. Compared to $D^n$,
the degree of the univariate polynomial from the
usual Kronecker substitution, we see a significant reduction when
$T \ll D^{n-1}$.

\begin{cor}
  Let $f\in\F[x_1,\ldots,x_n]$ with max degree less than $D$ and
  at most $T$ nonzero terms.

  For any constant error probability $0<\mu<1$, and
  integers $s_1,\ldots, s_n$ chosen randomly as in Lemma~\ref{lem:lambda},
  the polynomial
  $g(z) = f\left(z^{s_1}, z^{s_2},\ldots, z^{s_n}\right)$
  has degree at most $\oh{TD}$.
\end{cor}

\section{Multivariate diversification}
\label{sec:diverse}

Consider $f\in\F[x_1,\ldots,x_n]$ as in \eqref{eqn:f-mvar}:
$$
  f = a_1 \x^{\e_1} + a_2 \x^{\e_2} + \cdots + a_t \x^{\e_t}.
$$

Each choice of $\s = (s_1,\ldots,s_n)$ 
for a randomized Kronecker substitution maps $f$ to
a univariate polynomial $g(z) = f(z^{s_1},\ldots,z^{s_n})$. In order
to recover the exponent tuples of the original polynomial $f$, it will
be necessary in the next section to perform multiple such substitutions
and correlate terms in each $g$ that correspond to the same unknown term
in $f$.

The notion of \emph{diversification}, introduced by \citet{GR11a}, will
be used to correlate terms in the substituted polynomials $g$. The basic
idea is that distinct terms in $f$ will be made, through a
randomization, to have distinct coefficients.

In fact, as there will be some small number of collisions in each
substituted polynomial $g$, we require the notion of \emph{generalized
diversity} from \cite{AGR14}. The idea is that not only must the terms in $f$
have distinct coefficients, but some small number of \emph{sums} of
terms in $f$ must additionally be distinct.

This problem of diversification 
is to choose $\alpha$ from a suitable set so that, with high
probability, $\alpha$ is not a root of any in a set $\HH$ of polynomials.
In the original notion of diversity, $\HH$ simply consists of the set of
$\tfrac{T(T-1)}{2}$ pairwise term differences from $f$. To achieve
generalized diversity, we must also consider polynomials $h$ of the form
$$h = \sum_{i \in S} a_i \x^{\e_i} - \sum_{j \in S'}a_j\x^{\e_j},$$
where $S$ and $S'$ each either comprise a single term or a set of terms appearing in a collision.

\begin{definition}[Diversifying set]\label{def:divset}
  Let $n\ge 1$ and bounds $D$, $m$, and $\mu$ be given, and let
  $\HH\subseteq\F[x_1,\ldots,x_n]$ be an unknown subset of 
  nonzero polynomials satisfying $\#\HH\le m$ and 
  with each $h\in\HH$ having max degree less than $D$.

  We say $A\subseteq\F$ is a $(n,D,m,\mu)$-\emph{diversifying set} if
  the probability is less than $\mu$ that
  any evaluation point $\alvec$, with entries chosen at random from $A$,
  is a root of \emph{any} of the $h\in\HH$. That is,
  $$
    \Pr_{\alvec\in A^n}\left[h(\alvec) \ne 0\ \forall h\in\HH\right] \ge
    1-\mu.
  $$
\end{definition}

From the discussion above, we see that the set $\HH$
the differences between any of the $t \leq T$ single terms and
any of the $r$ sets of collisions, which is at most
$\#H < \tfrac{1}{2}(T+r)^2$.

\begin{lemma}
  Let $f\in\F[x_1,\ldots,x_n]$ with degree less than $D$ and at most $T$
  nonzero terms, and $0<\mu<1$.
  If there are at most $r$ collisions in some set of
  evaluations of $f$, and $\alvec\in\F^n$ is chosen at random from a set
  $A\subseteq\F$ that is a 
  \mbox{$(n,D,\tfrac{1}{2}(T+r)^2,\mu)$-diversifying} set, then, with probability at
  least $1-\mu$, every coefficient of $f(\alpha_1 x_1,\ldots, \alpha_n
  x_n)$ is distinct from every other coefficient and from the
  coefficients of all $r$ collisions.
\end{lemma}
\begin{proof}Follows from the discussion above and the definition of a
diversifying set.
\end{proof}

From the definition of diversifying set, a simple application of the
Schwartz-Zippel lemma could be used to generate diversifying sets as
long as the field $\F$ is sufficiently large.
Theorems 3.1 and 4.6 in \citet{GR11a} define diversifying sets for large
finite fields and fixed-precision complex numbers, respectively, in the univariate
case $n=1$. Our more recent work in \citet{AGR14} constructs
smaller diversifying sets by choosing \emph{vectors} of substitutions,
again for the univariate case. We restate these results in our current
notation and refer the reader to the aforementioned results for further
details.

\begin{cor}[\citet{GR11a}, Theorem 3.1]
  Let bounds $D$, $m$, and $\mu$ be given. If $q$ is a prime power
  satisfying $q \ge mD/\mu$, then the set $\FF_q^*$ of all nonzero
  elements in the finite field of size $q$ is a
  $(1,D,m,\mu)$-diversifying set.
\end{cor}

\citet{GR11a} also considers the case when $f \in \mathbb{C}[x]$, $f$ is
given by a numerical black box. Their proof does not apply here as the
polynomials in $\HH$ for us are not always binomials. We hope that a
similar result would hold for multivariate diversification, but do not
consider the question here.

\begin{cor}[\citet{AGR14}, Lemma 4.1]
  Let bounds $D$, $m$, and $\mu$ be given as above, $q$ a prime
  power, and set
  $$
    s = \ceil{\log_q (2D+1)} 
    \quad\text{and}\quad
    k = \ceil{\log\tfrac{2}{\mu} + 2\log m}.
  $$
  Then the set $\FF_{q^s}^k$ of $k$-tuples from a size-$s$ extension of
  the finite field with $q$ elements is a
  $(1,D,m,\mu)$-diversifying set%
  \footnote{In this case we make the abuse of notation that each
  evaluation $f(\alpha x)$ is actually a $k$-tuple of evaluations,
  and the coefficients in $f(\alpha x)$ are actually $k$-tuples in
  $\FF_{q^s}$.}%
  .
\end{cor}

We show, more generally, that vectorization may be applied to any diversifying set.

\begin{lemma}
  Let $n, D, m, \mu$ be given and $\HH$ an unknown set of polynomials
  as in Definition~\ref{def:divset}.  If $A \subset \F$ is a 
  $(n,D,m,\mu)$-diversifying set, then $A^\ell \subset \F^\ell$ is a 
  $(n,D,m,\mu^\ell)$-diversifying set, where addition and multiplication in $\F^\ell$ are component-wise.
\end{lemma}

\begin{proof}
As $A$ is a $(1,D,m,\mu)$-diversifying set, then by definition, a
randomly selected row vector $\alvec \in A^n$ satisfies $h(\alvec) \neq
0$ for all $h \in \HH$ with probability at least $1-\mu$.  Suppose
$\btvec=(\alvec_1,\ldots, \alvec_\ell)$ 
is chosen randomly from $A^{\ell \times n}$, and note that
$h(\btvec) = \big(h(\alvec_1),\ldots, h(\alvec_\ell)\big)$.
Thus the probability that $h(\btvec)=\zvec$ is the probability that 
$h(\alvec_i)=0$ for every $i=1,\ldots,\ell$, which is at most $\mu^\ell$.
\end{proof}

Rather than rehash the univariate diversification procedures, we refer
the reader to the aforementioned results and provide the following
connection which shows that univariate diversifying sets, with success
probability scaled by a factor of $n$, become multivariate diversifying
sets.

\begin{theorem}\label{thm:mult-diverse}
  Let $\F$ be an integral domain and $n, D, m, \mu$ be given and $\HH$ an unknown set of polynomials
  as in Definition~\ref{def:divset}.  If $A \subset \F$ is a 
  $(1,D,m,\mu/n)$-diversifying set, then $A$ is also a 
  $(n,D,m,\mu)$-diversifying set.
\end{theorem}
\begin{proof}
  The proof is by induction on $n$. When $n=1$, the
  statement holds trivially. So assume $n\ge 2$ and
  also that any $(1,d,m,\mu/(n-1))$-diversifying set is also a 
  $(n-1,d,m,\mu)$-diversifying set.

  We know that $\HH$ is a set of $n$-variate polynomials, each
  with max degree less than $D$. Rewrite each $h$ as a
  polynomial in $x_1$ with coefficients in $\F[x_2,\ldots,x_n]$, and define 
  $\lc(h)$ to be the leading coefficient of $h$ in terms of $x_1$.
  Each $\lc(h)$ is an $(n-1)$-variate polynomial with max degree less
  than $D$. Furthermore, every $\lc(h)$ must be nonzero since $h\ne 0$.

  Now define $\HH' = \{\lc(h) \mid h\in\HH\}$. This is a set of
  $(n-1)$-variate polynomials with degrees less than $D$. Therefore,
  by the induction hypothesis, a random
  evaluation point $(\alpha_2,\ldots,\alpha_n)$ with elements chosen
  from $A$ is a zero for any polynomial in $\HH'$ with probability less
  than $(n-1)\mu/n$.

  Now consider the set
  $\HH'' = \{h(x_1, \alpha_2, \ldots, \alpha_n) \mid h\in\HH\}$, whose
  leading coefficients are all nonzero with probability at least
  $1-\mu/(n-1)$. $\HH'$ is a set of at most $m$ univariate polynomials
  with degrees less than $D$. From the original definition of $A$,
  choosing $\alpha_1$ at random from $A$ makes every
  $h(\alpha_1,\alpha_2,\ldots,\alpha_n)$ nonzero with probability at
  least $1-\mu/n$.

  Therefore a randomly-chosen point $(\alpha_1,\ldots,\alpha_n)\in A^n$
  is a root of any $h\in\HH$ only if $(\alpha_2\ldots,\alpha_n)$ is a
  root of some polynomial in $\HH'$, or $\alpha_1$ is a root of some
  polynomial in $\HH''$. As the probability of each of these is less than
  $(n-1)\mu/n$ and $\mu/n$, respectively, the probability either occurs
  must be less than $\mu$, as required.
\end{proof}

\section{Multivariate interpolation}
\label{sec:interp}

In this section we show how one can interpolate $f$ using randomized
substitutions and a univariate interpolation algorithm. As in
\eqref{eqn:f-mvar}, write
$$
f = a_1\x^{\e_1} + \dots + a_t\x^{\e_t} \in \F[x_1, \dots, x_n].
$$
with known bound $T \geq t$. If $T=1$, then we simply
perform $n$ substitutions $f(z,1,\ldots,1)$, 
\ldots, $f(1,\ldots,1,z)$, each of which reveals the single term and its
exponent in one of the variables. No randomization is necessary in this
case and the solution is trivial.
Therefore for the remainder of this section we assume that
$T\ge 2$.

\subsection{Choosing multiple substitutions}
\label{sec:fam-subs}

The first step in interpolating $f$ is to select $\nu$ randomized
Kronecker substitutions, $\s_1, \ldots, \s_\nu \in \R^n$, where
each $s_i = (s_{i1},\ldots,s_{in})$.  
We require the $\s_i$ to be chosen in such a way that, 
with high probability,
every term of $f$ avoids collision for at least half of
the substitutions $\s_i$.

To achieve this we first randomly select Kronecker substitutions such that any fixed term of $f$ avoids collision for the substitution $\s$ with probability exceeding $3/4$.

For the bivariate case, we would choose primes 
$s_{i1} \in [\lambda_p, 2\lambda_p]$, and 
$s_{i2} \in [\lambda_q, 2\lambda_q]$, for $1 \leq i \leq
\nu$, where $\lambda_p$ and $\lambda_q$ are determined by setting
$\mu=1/4$ in Lemma \ref{lem:lambda}.  

Applying Lemma \ref{lem:lambda} for the general multivariate case, we
would select each integer $s_{ij} \in [0, \lambda-1]$,  
where $\lambda$ is the least
prime greater than $4T/3$.

Given such choices of $\s_i$, the following lemma shows how many
substitutions $\nu$ are required so that every term of $f$ appears
without collisions in at least half of them.

\begin{lemma}\label{lem:nu}
  Let $f\in\F[x_1,\ldots,x_n]$ with max degree less than $D$ and at most
  $T$ nonzero terms.
  Set $$\nu = \max\left(4n, 8\ln\left(10 T\right)\right),$$ and choose $\nu$ vectors
  $\s\in\ZZ^n$ such that, for any single $\s$ and any particular term in
  $f$, the probability that the term collides with another is less than
  $1/4$.

  Then, with probability at least $9/10$, every term of $f$
  collides with no others for at least $2n$ of the substitutions.
\end{lemma}
\begin{proof}
By Hoeffding's
inequality \citet{HOE63} the probability that any fixed term of $f$
collides in a proportion of at least $1/2$ of the substitutions is at
most $\exp(-\nu/8) \le 1/(10 T)$.  Thus the probability is at most
$1/10$ that any term of
$f$ collides in more than $\nu/2\ge 2n$ of the substitutions $\s$.
\end{proof}

\subsection{Choosing a diversifying set}
The next step is to find an appropriate diversifying set for the
interpolation problem.  Note the images 
$$g'_i = f(z^{s_{i1}}, \dots, z^{s_{in}}), 1 \leq i
\leq \nu,$$ 
contain at most $T\nu$ nonzero terms in total.  Every term $cz^d$ from
one of the images $g'_i$, $1 \leq i \leq \nu$, is the image of the
(possibly empty) sum of terms of $f$ of degree $\e$ satisfying
$\e\cdot\s_i = d$.  We take $\mathcal{H}$ to be the set of nonzero
differences of all such sums.
Thus, in order to obtain the
appropriate diversity with some desired probability $1-\epsilon$, we require a
$(n,D,\#\HH,\epsilon)$-diversifying set.  Per Theorem
\ref{thm:mult-diverse}, it suffices that we find a
$(1,D,\#\HH,\epsilon/n)$-diversifying set $A$.

We randomly select $\alvec = (\alpha_1, \dots, \alpha_n)$ from $A^n$ and then use a univariate interpolation algorithm of our choosing in order to construct the set of images
\begin{equation*}
g_i = f(\alpha_1z^{s_{i1}}, \ldots, \alpha_nz^{s_{in}}), \quad 1 \leq i \leq \nu
\end{equation*}
having the property that, with high probability, every
pair of terms $cz^d$ of $g_i$ and $cz^e$ of $g_j$ ($1 \leq i,j \leq
\nu$), sharing a coefficient $c$, are images of the same sum of terms of
$f$.

\begin{lemma}\label{lem:m}
  Let $f\in\F[x_1,\ldots,x_n]$ with max degree less than $D$ and at most
  $T$ nonzero terms.
  Assume $\nu$ substitution vectors are chosen according to
  Lemma~\ref{lem:nu}.

  Set $m=T^2(\nu+2)^2/8$ and choose
  a $(n,D,m,1/10)$-diversifying set $A \subseteq \F$.
  Then, with probability at least $4/5$, any nonzero coefficient $c$
  that appears in at least $\nu/2$ of the substitution polynomials $g_i$
  is the image of a single term in $f$.
\end{lemma}
\begin{proof}
  From the proof of
  Lemma~\ref{lem:nu}, the probability is at least $9/10$ that
  every term in $f$ appears without collision in at least $\nu/2$ of the
  substitutions. 
  
  Assuming this is the case, there can be at most $T\nu/4$ sums of terms
  that collide in any image $g_i$, since each collision involves at
  least two terms, there are at least $T\nu/2$ terms that do not
  collide, and the total number of terms in all images is $T\nu$.

  Hence the set of term differences $\HH$ will consist of the
  differences of any pair in a set of $T + T\nu/4$ polynomials. The
  number of such pairs is less than $m$ given in the statement of the
  lemma.

  From the definition of $A$, the
  probability that any of these polynomials in $\HH$ vanish on
  $\alvec\in A^n$
  is less than $1/10$, so the total probability that each term in $f$
  is uninvolved in collisions in at least $\nu/2$ of the images, and
  all distinct terms and collisions in the image polynomials $g_i$ have
  distinct coefficients, is at least
  $(9/10)^2 > 4/5$.
\end{proof}

A direct consequence is that any fixed subset of terms of $f$ must collide
in fewer than half of the $g_i$. For every nonzero coefficient that
occurs in at least $\nu/2$ of the images $g_i$, 
we know those terms with coefficient $c$ are probably images of the
same fixed term of $f$.

An alternate method might be to allow the
$O(T\nu)$ sums of terms that appear in collisions to sometimes share the
same coefficient, as long as these coefficients are not the same as any
of the $T$ coefficients of actual terms in 
$f(\alpha_1 x_1, \ldots, \alpha_n x_n)$. This would
reduce the $m$ in determining the diversifying set to
$T^2(\nu+2)/4$, a factor of $n$ improvement from the bound above.
The cost of such weakened diversifying sets would
be that some number $\le T/4$ of terms in the final recovered polynomial $h$ are not
actually terms in $f$. By iterating $O(\log T)$ times, such ``garbage
terms'' could be eradicated. 

\subsection{Recovering the multivariate exponents}
\label{sec:lin-sys}

For each coefficient $c$ that appears in at least $\nu/2$ of the images
$g_i$, we attempt to find $n$ linearly independent
substitution vectors, call them 
$\rr_1,\ldots,\rr_n \in \{\s_1,\ldots, \s_\nu\}$,
such that
every substitution polynomial $g_j$ with substitution vector $\rr_j = (r_{j1}, \dots, r_{jn})$, 
for $1\le j \le n$,
 contains the coefficient $c$ in a nonzero term.  
 
In the
bivariate case this is straightforward. Any $2 \times 2$
linear system formed by two
substitution vectors 
$$\begin{bmatrix}s_{11} & s_{12} \\ s_{21} & s_{22}\end{bmatrix}$$
must have nonzero determinant since in the bivariate case the entries
are all distinct prime numbers.

The general multivariate case is more involved, as $n$ substitution
vectors may not always be linearly independent.  For this case we will
randomly select $2n$ vectors $\rr_1,\ldots,\rr_{2n} \in [0,\lambda-1]^n$,
from which
we will search for $n$ linearly independent vectors. To that end we
require a bound on the probability that such $n$ independent vectors
do not exist.

\begin{lemma}\label{lem:linsys}
Let $\lambda$ be a prime number and
$f\in\F[x_1,\ldots,x_n]$,  and suppose that a
term of $f$ avoids collision in an image $f(\x^\s)$ for a randomly
chosen $\s \in [0,\lambda-1]^n$ with probability at least $3/4$.  Let
$\rr_1,
\ldots, \rr_{2n}$ be row vectors, chosen uniformly from
$[0,\lambda-1]^n$.  Given that a term of $f$ avoids collision in the images
$f(\x^{\rr_i})$ for $1 \leq i \leq 2n$, then

$$
Q = \begin{bmatrix} \rr_1 \\ \vdots \\ \rr_{2n} \end{bmatrix}
$$
has rank less than $n$ with probability at most $(9\lambda/16)^{-n}$.
\end{lemma}

\begin{proof}
Since all entries in each $\s$, and thereby everything in each $\rr$ and
every element in $Q$, is less than $\lambda$, we can consider all these
objects over the finite field $\FF_\lambda$. 

If $Q$ has rank less than $n$, then $\rr_1, \dots \rr_{2n}$  all lie in some
dimension-$(n-1)$ subspace $W \subset \FF_\lambda^n$. 
The number of distinct substitution vectors that could lie in the 
same subspace $W$ is $\lambda^{n-1}$.

Each dimension-$(n-1)$ subspace $W \subset \FF_\lambda^n$ may be specified by a
nonzero vector spanning its orthogonal space, unique up to a scalar
multiple.  Thus the number of such subspaces is less than $\lambda^n$, and
so the number of possible $2n$-tuples comprised of substitution vectors
that do not span $\FF_\lambda^n$ is at most
$$
\lambda^n (\lambda^{n-1})^{2n} = \lambda^{2n^2 - n}.
$$

Meanwhile, there are $\lambda^{2n^2}$ possible $2n$-tuples of
substitution vectors, and so the probability that such a tuple does span
$V$ is at most $\lambda^{-n}$.  Furthermore, by the hypothesis, the probability that a term of $f$ avoids collision for each substitution $\rr_i$ is $(3/4)^{2n}$, and thus the conditional probability that $Q$ is not full rank given that a fixed term of $f$ avoids collision for each $\rr_i$ is at most $\lambda^{-n}/(3/4)^{2n} = (9\lambda/16)^{-n}$.
\end{proof}

Given such a high probability of each term producing a rank-$n$ system
of substitution vectors, it is a simple matter to show that with high
probability \emph{every} term of $f$ admits some such rank-$n$ linear
system of substitutions without collisions.

\begin{cor}
  Let $f\in\F[x_1,\ldots,x_n]$ as above, and set $\nu$ according to 
  Lemma~\ref{lem:nu}, $\alvec$ according to Lemma~\ref{lem:m},
  and $\lambda \geq 3$ according
  to Lemma~\ref{lem:lambda}. With probability at least $2/3$, for
  every term in $f$, there
  exists a rank-$n$ set of substitution vectors $\rr_1,\ldots, \rr_n$, 
  such that the given term of $f$ does not collide in any of the
  substitutions $g = f(\alpha_1 z^{r_{i1}},\ldots \alpha_n z^{r_{in}})$,
  for $1\le i\le n$.
\end{cor}
\begin{proof}
  As we have discussed, the case $T=1$ is trivial and when $n=2$ we
  choose primes for the vectors $\s$ and the $2\times 2$ linear systems
  always have full rank.

  So assume $T\ge 2$ and $n\ge 3$.
  We know from Lemma~\ref{lem:linsys} that
  the probability of a single term \emph{not} admitting a rank-$n$
  system of substitution vectors is less than $(9\lambda/16)^{-n}$,
  so the probability that any term does not have a rank-$n$ system of
  non-colliding substitution vectors is less than
  $T/(9\lambda/16)^n$. 
  
  Since $\lambda$ is chosen as the least odd prime greater than
  $4T/3$, we see that 
  $T/(9\lambda/16)^n \ge \tfrac{4}{3}(9\lambda/16)^{-n+1}$. And because $\lambda\ge 3$
  and $n\ge 3$, $\tfrac{4}{3}(9\lambda/16)^{-n+1} \le \tfrac{4}{3}(27/16)^{-2}<1/6$.

  Combining this with the probability bound from Lemma~\ref{lem:m}, the
  overall success probability is at least 
  $\tfrac{5}{6}\cdot\tfrac{4}{5} = \tfrac{2}{3}$.
\end{proof}

From the set of $2n$ vectors $\rr_1, \dots, \rr_{2n}$ we can find
$n$ linearly independent vectors by inspection of the $LU$ factorization
of the matrix whose $2n$ rows are the $\rr_j$.  By \citet{BH74}, we
can do this in $\softoh{n^\omega}$ operations in $\FF_\lambda$, for a
bit cost of $\softoh{n^\omega\log T}$.  We suppose by reordering
of the $\rr_j$ that $\rr_1, \dots, \rr_n$ are our $n$
linearly independent vectors.

Then, if $d_j$ is the exponent of the term with coefficient $c$ appearing in 
$f(\alpha_1z^{r_{j1}}, \ldots, \alpha_nz^{r_{jn}})$,
we may find the degree $\e$ of the term with coefficient $c$
in the diversified multivariate polynomial $f(\alpha_1x_1, \dots, \alpha_nx_n)$
by way of the linear system
\begin{equation*}
\begin{bmatrix} \rr_1 \\ \vdots \\ \rr_n \end{bmatrix}\begin{bmatrix} e_1 \\ \vdots \\ e_n \end{bmatrix} = \begin{bmatrix} d_1 \\ \vdots \\ d_n \end{bmatrix}.
\end{equation*}

We construct and solve such a linear system for every term of $f$,
giving us the polynomial
$$g=f(\alpha_1x_1, \dots, \alpha_nx_n),$$ from which we
easily obtain $f$ as $f=g(\alpha_1^{-1}x_1, \dots, \alpha_n^{-n}x_n)$.  Procedure \ref{proc:interp} describes the approach laid out in sections \ref{sec:fam-subs}-\ref{sec:lin-sys}.

\begin{procedure}[ht!]
\caption{Interpolate($f, n, T, D$) }\label{proc:interp}
\label{proc:FindPrimes}
\KwIn{Bounds $D, T$ and a black box for evaluating 
$f \in \F[x_1,\dots,x_n]$, an unknown polynomial with 
partial degrees less than $D$ and at most $T$ nonzero terms.  
}
\KwResult{
We construct $f$ with probability at least $2/3$.
}
\medskip

\tcp{Choose substitution vectors}

$\nu \gets \max(4n, 8\ln(10T))$ \;
\If{$n=2$}{
  $\s_1,\ldots,\s_\nu \gets$ prime vectors chosen by
  Lemma~\ref{lem:lambda-pq}\;
}\Else{
  $\s_1\ldots,\s_\nu \gets$ integer vectors chosen by
  Lemma~\ref{lem:lambda}\;
}
\medskip

\tcp{Diversify}
$m \gets T^2(\nu + 2)^2/8$\;
$A \gets$ a $(n,D,m,1/10)$-diversifying subset of $\F$\;
$\alvec \gets$ an element of $A^n$ chosen uniformly at random\;
\medskip

\tcp{Build images of $f$}
\For{$i=1,\ldots, \nu$}{
	$g_i \leftarrow f(\alpha_1 z^{s_{i1}}, \dots, \alpha_n z^{s_{in}})$
    via univariate\\
    \hspace*{7mm} interpolation
}\medskip

\tcp{Reconstruct $f$ from its images}
$g \gets 0$\;
\ForEach{coefficient $c \neq 0$ appearing in any $g_i$}{
  $S\gets \{\s_i \mid c\text{ is a coefficient in } g_i\}$\;
  \lIf{$\#S < \nu/2$}{{\bf continue}}
  $\rr_1,\ldots, \rr_n \gets$ linearly independent vectors chosen\\
  \hspace*{18 mm} from the first $2n$ vectors in $S$\;
  $d_1,\ldots, d_n \gets$ degrees of terms with coefficient $c$\\
  \hspace*{19mm}under substitutions $\rr_1, \dots, \rr_n$\;
  $R \gets (\rr_1\ \ldots\ \rr_n)^\top$\;
  $\dd \gets (d_1 \dots d_n)^\top$\;
  $\e \gets R^{-1}\dd$\;
  $g \gets g + c\x^\e$
}
\KwRet{$g(\alpha_1^{-1}x_1, \ldots, \alpha_n^{-1}x_n)$}

\end{procedure}

\subsection{Cost analysis}

We can now state the tangible benefit of our new Kronecker substitution
technique. The diversifying sets are included in these theorems even
though we do not actually count the cost of possibly extending the ring
$\R$ to include such sets.

\begin{theorem}
  For given bounds $D_x,D_y,T$
  and an unknown polynomial $f\in\R[x,y]$ with
  $\deg_x f < D_x$, $\deg_y f < D_y$, and $\#f \le T$,
  Procedure~\ref{proc:interp} succeeds in finding $f$ with
  probability at least $2/3$ and requires
  \begin{compactitem}
    \item $O(\log T)$ calls to univariate interpolation with
      $T$ nonzero terms and degree
      $\oh{\sqrt{T}\sqrt{D_x D_y}\log(D_x D_y)}$,
    \item A $(2,D,\oh{T^2\log^2 T),1/10}$-diversifying set in $\R$, and
    \item $\softoh{T\log D + \log^2 D}$ additional bit operations,
      where $D = \max(D_x,D_y)$.
  \end{compactitem}
\end{theorem}
\begin{proof}
  Since $n=2$, we have $\nu \in \oh{\log T}$.
  This is the number of calls to the univariate interpolation algorithm,
  and the degree bound comes from Corollary~\ref{cor:deg-bivar}.
  The size of $\HH$ in the diversifying set comes from the fact that
  $m \in \oh{T^2\log^2 T}$.

  Two steps dominate the bit complexity. First, we must choose $2\nu$
  primes in $[\lambda,2\lambda]$. This can be accomplished via
  $O(\nu)$ applications of the Miller-Rabin primality test, performing
  $O(\log\log T)$ trials each time to ensure a negligible probability
  of error. The cost of these tests is $\softoh{\log T\log^2\lambda}$,
  which is $\softoh{\log T \log^2 D}$.

  The other dominating step in bit complexity is simply the cost of
  computing with the $T\nu$ exponents in images $g_i$ and $T$ exponent
  vectors in the final result. There are $\softoh{T}$ such exponents,
  each with $O(\log D)$ bits, for a total bit cost of 
  $\softoh{T\log D}$.
\end{proof}

\begin{theorem}
  For given bounds $D,T$
  and an unknown polynomial $f\in\R[x_1,\ldots,x_n]$ with
  $\max\deg g < D$ and $\#f \le T$,
  Procedure~\ref{proc:interp} succeeds in finding $f$ with
  probability at least $2/3$ and requires
  \begin{compactitem}
    \item $O(n + \log T)$ calls to univariate interpolation with
      $T$ nonzero terms and degree $\oh{TD}$,
    \item A $(2,D,\softoh{T^2n^2},1/10)$-diversifying set in $\R$, and
    \item $\softoh{n^\omega T + nT\log D}$ additional bit
    operations, where $2 < \omega < 3$ is the exponent of matrix
    multiplication.
  \end{compactitem}
\end{theorem}
\begin{proof}
  The analysis of the first two parts is the same as in the bivariate
  case.

  For the bit complexity, we do not have to worry about primality
  testing here. However, the size of all exponents in the polynomials
  becomes $\softoh{nT\log D}$, and the cost of performing each LU
  factorization on a $(2n)\times n$ matrix is $O(n^\omega)$ operations
  on integers with $O(\log T)$ bits. As $T$ such LU factorizations are
  required, the total bit cost of the linear algebra is
  $\softoh{n^\omega T}$.
\end{proof}

\subsection{Interpolating $f$ with arbitrarily high probability}

Interpolating $f$ entails the probabilistic steps of (1) selecting a set
of randomized substitutions that produce few collisions
and (2) selecting $\alvec$ from a diversifying set $A$ such that 
all term sums in all images have distinct coefficients in those
images.
The $n\ge 3$ case has in addition the
probabilistic step (3) of guaranteeing that we can construct a full-rank
linear system in order to solve for every exponent of $f$.  
The probability of failure in each of these steps has been controlled
above so that the overall success probability is at least $2/3$.

If a higher success probability, say $1-\epsilon$, is desired, 
we simply run the interpolation algorithm described in sections
\ref{sec:fam-subs}--\ref{sec:lin-sys} with some $\ell$
times.  Again using Hoeffding's inequality, the probability that the
algorithm fails at least $\ell/2$ times is at most 
$\exp(-2\ell(1/6)^2) = \exp(-\ell/18)$. Thus, if we wish
to discover $f$ with probability $1-\epsilon$, we
merely run the algorithm as suggested some $\ell=\lceil 18\ln
\tfrac{1}{\epsilon}\rceil$ times, and select the 
the polynomial $f$ that is returned a majority of the time.
With probability at least $1-\epsilon$, such an $f$ exists and is in
fact the correct answer.

\section{Perspective}\label{sec:persp}

We have presented a new randomization that maps a multivariate
polynomial to a univariate polynomial with (mostly) the same terms. This
improves on the usual Kronecker map by reducing the degree of the
univariate image when the polynomial is known to be sparse. We have also
shown how a small number of such images can be combined to recover the
original terms of the unknown multivariate polynomial.

There are numerous questions raised by this result. Perhaps foremost is
whether there is any practical gain in any particular application by
using this approach. We know that the randomized Kronecker substitution
will result in smaller degrees than the usual Kronecker substitution
whenever the polynomial is sufficiently large and sufficiently sparse,
so in principle the applications should include any of the numerous
results on sparse polynomials that use a Kronecker substitution to
accommodate multivariate polynomials.

Unfortunately, in practice, the situation is not so clear. Many of the
aforementioned results that rely on a Kronecker substitution either do
not have a widely-available implementation, or do not usually involve
sparse polynomials. However, for the particular applications of sparse
GCD and sparse multivariate multiplication, there is considerable
promise particularly in the case of bivariate polynomials with degree
greater than 1000 or so and sparsity between $D$ and $D^2$. An efficient
implementation comparison in these situations would be useful and
interesting, and we are working in that direction.

There are also questions of theoretical interest. For one, we would
like to know
how far off the bounds on the size of primes from Lemmata
\ref{lem:lambda-pq} and \ref{lem:lambda} are compared to what is really
necessary to avoid collisions. 

An important question is whether our current results are optimal in any sense.
In the bivariate case, 
when $D_x = D_y$ our result gives $p,q \in \softoh{\sqrt{T}}$, which is
optimal in terms of $T$. That is because $T$ could be as large as $\Theta(D^n)$,
and therefore any monomial substitution exponent less than 
$\Omega(T^{(n-1)/n})$ would by necessity have more than a constant fraction
of collisions.
However our result for $n\ge 2$ gives each
$p_i\in \softoh{T + n^2 + \log^2 D}$, which in terms of $T$ is off by a
factor of $T^{1/n}$ from the optimal. It may be possible to improve
these bounds simply with a better analysis, or with a different kind of
randomized monomial substitution. In either case, it is clear that, at
least for $n\ge 3$ and in particular for $n=3$, it should be possible to
improve on the results here and achieve univariate reduced polynomials
with even lower degree.

Another interesting question would be whether some of this randomization can be
avoided. Here we have two randomizations, the diversification and the (multiple)
randomized Kronecker substitutions. And this is besides any randomization that
might occur in the underlying univariate algorithm! It seems plausible that, for
example in the application of multivariate multiplication, the known aspects of
the monomial structure might be used to make some choices less random and more
``intelligent''. However, we do not yet know any reasonable way to accomplish
this.

\section{Acknowledgements}

We wish to thank Zeev Dvir for pointing out the previous work of
Klivans and Spielman, and the reviewers for their helpful comments.
We also thank the organizers of the SIAM AG13 meeting 
for the opportunity to discuss preliminary work on this topic.
The second author is supported by the National Science Foundation,
award \#1319994.

A version of this paper will appear at ISSAC 2014 in Kobe, Japan.

\newcommand{\Gathen}{\relax}\newcommand{\Hoeven}{\relax}

\end{document}